\documentclass[10pt]{article}
\usepackage{tikz} 
\newcommand{\blackcell}{\cellcolor{black} \color{black} 0000} 
\usepackage{colortbl} 

\usepackage[utf8]{inputenc} 
\usepackage{geometry} 
\usepackage{amsmath} 
\usepackage{amsthm} 
\usepackage{amssymb} 
\newtheorem{definition}{Definition}
\newtheorem{note}{Note}

\usepackage{bm} 

\usepackage{authblk} 
\usepackage{graphicx} 
\usepackage{subcaption} 

\usepackage{listings} 
\usepackage{xcolor} 
\lstdefinestyle{CStyle}{ 
  language=C,
  basicstyle=\ttfamily,
  keywordstyle=\color{blue},
  commentstyle=\color{green!50!black},
  stringstyle=\color{red},
}

\usepackage{cancel} 
\usepackage{mathtools} 
\usepackage{enumitem} 

\usepackage{algpseudocode} 

\usepackage{hyperref} 
\hypersetup{ 
    colorlinks=true,
    linkcolor=blue,
    filecolor=magenta,      
    urlcolor=cyan,
    pdftitle={Overleaf Example},
    pdfpagemode=FullScreen,
}

\usepackage{appendix} 
\usepackage[english]{babel} 
\usepackage[nottoc]{tocbibind} 
\usepackage{fancyhdr} 
\usepackage[ruled,vlined,linesnumbered]{algorithm2e}
\usepackage[ruled,vlined]{algorithm2e} 

\makeatletter
\newcommand{\thickhline}{%
    \noalign {\ifnum 0=`}\fi \hrule height 2pt
    \futurelet \reserved@a \@xhline
}

\title{Digit-Indexed q-Ary SEC–DED Codes with Near-Hamming Overhead}

\author{Jiaxu Hu, Kenneth J. Roche}
\affil{University of Washington, Seattle, Washington 98195-1560, USA}

\date{\today}

\newcommand{\F}{\mathbb{F}}

\newtheorem{theorem}{Theorem}
\newtheorem{lemma}{Lemma}
\usepackage{booktabs}
\usepackage{tabularx}

\begin{document}

\maketitle

\begin{abstract}
We present a simple $q$-ary family of single-error-correcting, double-error-detecting (SEC--DED) linear codes whose parity checks are tied directly to the base-$p$ ($q{=}p$ prime) digits of the coordinate index. For blocklength $n=p^r$ the construction uses only $r{+}1$ parity checks---\emph{near-Hamming} overhead---and admits an index-based decoder that runs in a single pass with constant-time location and magnitude recovery from the syndromes. Based on the prototype, we develop two extensions: Code A1, which removes specific redundant trits to achieve higher information rate and support variable-length encoding; and Code A2, which incorporates two group-sum checks together with a 3-wise XOR linear independence condition on index subsets, yielding a ternary distance-4 (SEC--TED) variant. Furthermore, we demonstrate how the framework generalizes via $n$-wise XOR linearly independent sets to construct codes with distance $d = n + 1$, notably recovering the ternary Golay code for $n = 5$—showing both structural generality and a serendipitous link to optimal classical codes.

Our contribution is not optimality but \emph{implementational simplicity} and an \emph{array-friendly} structure: the checks are digitwise and global sums, the mapping from syndromes to error location is explicit, and the SEC--TED upgrade is modular. We position the scheme against classical $q$-ary Hamming and SPC/product-code baselines and provide a small comparison of parity overhead, decoding work, and two-error behavior.
\end{abstract}

\section{Introduction}
Error correction codes are fundamental to modern communication systems, ensuring data transmission accuracy by detecting and correcting errors introduced by noisy channels. Among these, the Hamming code is a classical solution capable of correcting any single error and detecting double errors in binary data. Extending Hamming codes to ternary systems takes advantage of the greater symbol efficiency of three-state logic (0, 1, 2). This enables denser data encoding in multi-level storage (e.g., NAND flash) and phase-modulated optical communications, where the binary alphabet is inherently limiting.\\

This paper presents a study on ternary error-correcting codes, with a focus on constructions derived from digit-indexed parity checks, exploring their unique properties in non-binary systems. Our work begins by constructing a prototype single-error-correcting, double-error-detecting (SEC--DED) code. By extending the parity-check framework, we then introduce modular arithmetic and the novel concept of "inverse index pairs", which leads to a more computationally efficient encoder and decoder. This framework is applied to create the A1 variant code, which achieves approximately 50\% higher information density than the prototype under identical redundancy. Furthermore, by incorporating a 3-wise linearly independent set of indices, we develop the A2 code—a ternary SEC--TED variant with an increased minimum distance of 4.\\

The paper is structured as follows. Section 2 presents the ternary extension of the Hamming code, including a three-dimensional SEC--DED prototype based on digit-indexed parity checks and its generalization to arbitrary prime bases. Section 3 formalizes modular arithmetic principles and the concept of inverse index pairs. Section 4 builds upon this foundation, exploring applications of the inverse-pair structure for variable block lengths, higher code rates, and increased minimum distances; theoretical proofs and algorithms for the A1 and A2 codes are provided to demonstrate their effectiveness. Finally, Section 5 develops a general method for constructing codes with larger minimum distances, using the ternary Golay code as an illustrative example.\\

By extending the Hamming code paradigm through digit-indexed parity and inverse-pair structures, this work provides a simple, efficient, and modular framework for non-binary error correction. We thus advance the design of ternary codes and open up new research directions in higher-dimensional coding. To situate our approach, we next outline applications motivating $q$-ary codes, then position our scheme against established families, provide a comparative summary, and finally state our specific contributions.

\paragraph{Applications (why $q$-ary)} 
Multi-level storage and modulation naturally operate over nonbinary alphabets: e.g., $q{=}3$ trits for ternary memories, $q{=}4$/$8$ symbol alphabets for flash/PCM, or $q$-ary constellations in coded modulation.
In these settings, SEC--DED with extremely simple parity structure and an index-based locator can reduce control logic, latency, and implementation cost.

\paragraph{Positioning and novelty} 
Our SEC--DED family for $n=p^r$ uses $r{+}1$ parities---close to the $r$ parities of the optimal $q$-ary Hamming codes at comparable lengths. 
Unlike Hamming, the parity rows here correspond directly to base-$p$ digits (plus one global sum), which produces a decoder whose locator is an explicit index read-off rather than a table or algebraic search; this also makes an SEC--TED upgrade possible by adding two structured ``group'' checks. 
Compared to SPC/product codes, our scheme has lower redundancy at these lengths and a cleaner single-pass decoder, though it forgoes some of the distance/iterative-decoding advantages of product constructions.
We therefore present this family as a clean implementation point with a modular path to $d{=}4$, not as a new optimal code family.

\paragraph{At-a-glance comparison}
Table~\ref{tab:compare} summarizes parity overhead, decoding work, and two-error behavior for three baselines.
(Here $N$ denotes the total length(block length), $R$ is number of redundant digits, and constants are indicative; exact costs depend on implementation details.)

\begin{table}[h]
\centering
\scriptsize
\setlength{\tabcolsep}{2pt}
\caption{SEC--DED landscape with implementation cost. Example: $N=3^5=243$ ($r=5$) for Proposed/Hamming; $m=16$ ($N=256$) for SPC$\times$SPC.}
\label{tab:compare}
\begin{tabularx}{\linewidth}{@{}%
  >{\raggedright\arraybackslash}p{.15\linewidth}%
  >{\centering\arraybackslash}p{.05\linewidth}%
  >{\centering\arraybackslash}p{.05\linewidth}%
  >{\raggedright\arraybackslash}p{.23\linewidth}%
  >{\raggedright\arraybackslash}p{.20\linewidth}%
  >{\raggedright\arraybackslash}p{.28\linewidth}@{}}
\toprule
Family & N & R & Decoding work & Distance: xEC--yED & Impl.\ cost (example) \\
\midrule
Proposed SEC--DED
& $p^r$
& $r{+}1$ 
& One-pass syndromes; explicit index locator (O(1)) 
& Distance 3: SEC--DED
& $\approx N(r{+}1)$ adds; e.g., $243\times6=1458$ adds; locator O(1) \\
\vspace{0pt}\\
$q$-ary Hamming 
& $\frac{q^r-1}{q-1}$
& $r$ 
& One-pass syndromes; algebraic/table locator (O(1)) 
& Distance 3: SEC--DED
& $\approx Nr$ adds; e.g., $243\times5=1215$ adds; locator O(1) \\
\vspace{0pt}\\
SPC $\times$ SPC 
& $m^2$ 
& $2m-1$ 
& Two passes (row \& column checks); coordinatized locator 
& Distance 4: SEC--TED
& Row $\approx mN$, col.\ $\approx mN$; e.g., $16\times256\times2\approx8192$ adds; locator O(1) \\
\bottomrule
\end{tabularx}
\end{table}

\paragraph{Contributions}
\hspace{0pt}\\
(i) A digit-indexed $q$-ary SEC--DED with $r{+}1$ parities for $N{=}p^r$ and an explicit index-based decoder.\\ 
(ii) A modular SEC--TED variant (ternary) via two group-sum checks and a 3-wise independence condition. \\
(iii) Formal parity-check presentations ($H$ and $H'$) and proofs of distance ($d{=}3$ and $d{\ge}4$).


\section{Prototype Ternary SEC--DED Code}
This study mainly relates to a type of ternary error correction code that is based on the concept of the Hamming code. We briefly introduce the elementary version here, which is a $[27,23,3]_3$ SEC--DED code.

\subsection{Definition and Description}
We consider 27 positions arranged in a $3 \times 3 \times 3$ cube or three $3 \times 3$ matrices with 27 trits(ternary digits). There are four redundant trits $v_{mod3},v_{x_2},v_{x_1},v_{x_0}$ that serve as error correction, located at positions $(0,0,0)$, $(0,0,1)$, $(0,1,0)$, and $(1,0,0)$, corresponding to the $0^{th}$, $1^{st}$, $3^{rd}$, and $9^{th}$ positions.

\begin{figure}[h!]
    \centering
    \begin{minipage}[b]{0.20\textwidth}
        \includegraphics[width=\textwidth]{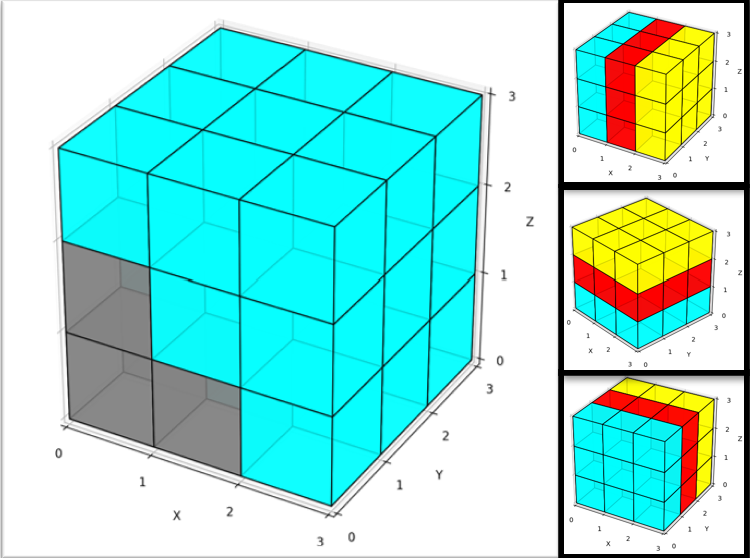}
        \caption{\scriptsize Example cube layout for ternary Hamming code.Blue indicate the $\cdot0$ region, red the $\cdot1$ region, and yellow the $\cdot2$ region.}
        \label{fig:ternary-hamming-cube}
    \end{minipage}
    \hfill
    \begin{minipage}[b]{0.75\textwidth}
        \centering
        \raggedright
        In the cube, we arrange the 23 message trits in the remaining trit blocks, which can take values $\{-1, 0, +1\}$(or equivalently $\{0,1, 2\}$). We now define the values of the redundant trits $v_{x_{0}},v_{x_{1}},v_{x_{2}},v_{mod3}\in\{-1,0,1\}$:(also, $\cdot2$ is equivalent to  $\cdot-1$ under modulo 3)
        \begin{itemize}
            \item $v_{x_0}$ for $(0,0,1)$: s.t. $\sum{v_{(x_2,x_1,1)}} + 2 \cdot \sum{v_{(x_2,x_1,2)}} \equiv 0 \pmod{3}$.
            \item $v_{x_1}$ for $(0,1,0)$: s.t. $\sum{v_{(x_2,1,x_0)}} + 2 \cdot \sum{v_{(x_2,2,x_0)}} \equiv 0 \pmod{3}$.
            \item $v_{x_2}$ for $(1,0,0)$: s.t. $\sum{v_{(1,x_1,x_0)}} + 2 \cdot \sum{v_{(2,x_1,x_0)}} \equiv 0 \pmod{3}$.
            \item $v_{mod3}$ for $(0,0,0)$: s.t. $\sum{v_{(x_2,x_1,x_0)}} \equiv 0 \pmod{3}$.
        \end{itemize}
    \end{minipage}
\end{figure}
    This is how we encode the prototype SEC--DED code.\\
Note: In further part, you will see that this method is equivalent to calculating the sum of the product of the value of each message trits with its index under Addition mod 3. And then set the redundant trits to let such sum back to 0.\\

\subsection{Example}
Assume that we want to record a ternary message of length 23:
\textbf{20,111,020,010,201,200,120,012}.


\begin{algorithm}[H]
\caption{Encode Prototype SEC--DED Code($-1\equiv2$): 20,111,020,010,201,200,120,012}
\begin{algorithmic}[1]
    \State \quad \textbf{Initialize: }Place the message trits in the matrix entries in row-major, skipping indices 0 and power of $3^n$, where $n\in \mathbf{N}_0$, (i.e., 0,1,3,9 in this cases).
    \[
    \begin{bmatrix}
        v_{mod3} & v_{x_0} & -1 \\
        v_{x_1} & 0 & 1 \\
        1 & 1 & 0 \\
    \end{bmatrix}
    \quad
    \begin{bmatrix}
        v_{x_2} & -1 & 0 \\
        0 & 1 & 0 \\
        -1 & 0 & 1 \\
    \end{bmatrix}
    \quad
    \begin{bmatrix}
        -1 & 0 & 0 \\
        1 & -1 & 0 \\
        0 & 1 & -1 \\
    \end{bmatrix}
    \]
    
    \State \quad \textbf{Calculate and Set the  $v_{x_2}$: } Such that $v_{x_2} - 2 \equiv 0 \pmod{3}$. Then, $v_{x_2}=-1$.
    \[
    \begin{bmatrix}
        v_{mod3} & v_{x_0} & -1 \\
        v_{x_1} & 0 & 1 \\
        1 & 1 & 0 \\
    \end{bmatrix}
    \quad
    \textcolor{blue}{
    \begin{bmatrix}
        v_{x_2}=-1 & -1 & 0 \\
        0 & 1 & 0 \\
        -1 & 0 & 1 \\
    \end{bmatrix}
    }
    \quad
    \textcolor{red}{
    \begin{bmatrix}
        -1 & 0 & 0 \\
        1 & -1 & 0 \\
        0 & 1 & -1 \\
    \end{bmatrix}
    }
    \]
    
    \[
    1 \cdot (\textcolor{blue}{v_{x_2} + (-1) + 1 + (-1) + 1}) + 2 \cdot (\textcolor{red}{-1 + 1 + (-1) + 1 + (-1)}) = v_{x_2} - 2
    \]

    \State \quad \textbf{Calculate and Set the  $v_{x_1}$: } Such that $v_{x_1} + 6 \equiv 0 \pmod{3}$. Then, $v_{x_1} = 0$.
    \[
    \begin{bmatrix}
        v_{mod3} & v_{x_0} & -1 \\
        \textcolor{blue}{v_{x_1}=0} & \textcolor{blue}{0} & \textcolor{blue}{1} \\
        \textcolor{red}{1} & \textcolor{red}{1} & \textcolor{red}{0} \\
    \end{bmatrix}
    \quad
    \begin{bmatrix}
        -1 & -1 & 0 \\
        \textcolor{blue}{0} & \textcolor{blue}{1} & \textcolor{blue}{0} \\
        \textcolor{red}{-1} & \textcolor{red}{0} & \textcolor{red}{1} \\
    \end{bmatrix}
    \quad
    \begin{bmatrix}
        -1 & 0 & 0 \\
        \textcolor{blue}{1} & \textcolor{blue}{-1} & \textcolor{blue}{0} \\
        \textcolor{red}{0} & \textcolor{red}{1} & \textcolor{red}{-1} \\
    \end{bmatrix}
    \]

    \[
    1 \cdot (\textcolor{blue}{v_{x_1} + 1 + 1 + 1 +(-1)}) + 2 \cdot (\textcolor{red}{1 + 1 + (-1) + 1 + (-1)}) = v_y + 2 + 2\cdot2 = v_{x_1} +6
    \]

    \State \quad \textbf{Calculate and Set the  $v_{x_0}$: } Such that $v_{x_0} + 1 \equiv 0 \pmod{3}$. Then, $v_{x_0}=-1$.
    \[
    \begin{bmatrix}
        v_{mod3} & \textcolor{blue}{v_{x_0}=-1} & \textcolor{red}{-1} \\
        0 & \textcolor{blue}{0} & \textcolor{red}{1} \\
        1 & \textcolor{blue}{1} & \textcolor{red}{0} \\
    \end{bmatrix}
    \quad
    \begin{bmatrix}
        -1 & \textcolor{blue}{-1} & \textcolor{red}{0} \\
        0 & \textcolor{blue}{1} & \textcolor{red}{0} \\
        -1 & \textcolor{blue}{0} & \textcolor{red}{1} \\
    \end{bmatrix}
    \quad
    \begin{bmatrix}
        -1 & \textcolor{blue}{0} & \textcolor{red}{0} \\
        1 & \textcolor{blue}{-1} & \textcolor{red}{0} \\
        0 & \textcolor{blue}{1} & \textcolor{red}{-1} \\
    \end{bmatrix}
    \]

    \[
    1 \cdot (\textcolor{blue}{v_{x_0} + 1 + (-1) + 1 + (-1) +1}) + 2 \cdot (\textcolor{red}{-1 + 1 +1 +(-1)}) = v_{x_0} +1
    \]

    \State \quad \textbf{Calculate and Set the  $v_{mod3}$: } Such that $v_{mod3} - 1 \equiv 0 \pmod{3}$. Then, $v_{mod3}=1$.
    \[
    \begin{bmatrix}
        v_{mod3}=1 & -1 & -1 \\
        0 & 0 & 1 \\
        1 & 1 & 0 \\
    \end{bmatrix}
    \quad
    \begin{bmatrix}
        -1 & -1 & 0 \\
        0 & 1 & 0 \\
        -1 & 0 & 1 \\
    \end{bmatrix}
    \quad
    \begin{bmatrix}
        -1 & 0 & 0 \\
        1 & -1 & 0 \\
        0 & 1 & -1 \\
    \end{bmatrix}
    \]

    \[
    \sum{v_{(x_2,x_1,x_0)}} = v_{mod3} -1
    \]
    
\State \textbf{End:} \quad This gives us a ternary error‑correcting code that can correct up to one error. When it handles cases with double errors, it will detect that an error exists but incorrectly identify as a single‑error case. Therefore, we call it a single‑error‑correcting, double‑error‑detecting (SEC‑DED) code with a minimum Hamming distance of 3.
\end{algorithmic}
\end{algorithm}
So we get an error correction code: \textbf{122,001,110,220,010,201,100,120,012}. We convert 1 on the eighth trit to 2 as potential noise: \textbf{122,001,1\textcolor{green}{2}0,220,010,201,100,120,012}. The following are processes to decode such error code, fix and get the original message. 
    \[
    \begin{bmatrix}
        \bm{1} & \bm{-1} & -1 \\
        \bm{0} & 0 & 1 \\
        1 & 1 \rightarrow\textcolor{green}{ -1} & 0 \\
    \end{bmatrix}
    \quad
    \begin{bmatrix}
        \bm{-1} & -1 & 0 \\
        0 & 1 & 0 \\
        -1 & 0 & 1 \\
    \end{bmatrix}
    \quad
    \begin{bmatrix}
        -1 & 0 & 0 \\
        1 & -1 & 0 \\
        0 & 1 & -1 \\
    \end{bmatrix}
    \]
\begin{center}
    \textit{Note: Convert 1 on 021 to -1 as potential noise.}
\end{center}

\begin{algorithm}[H]
\caption{Decode Prototype SEC--DED Code: 122,001,120,220,010,201,100,120,012}
\begin{algorithmic}[1]
    \State \quad We start error correction with this code here, assuming there is at most one error.
    
    Define $E = (x_2,x_1,x_0)$ to be the location of the error, with change of the value $\Delta v_E$.

    \State \quad \textbf{Check the $v_{mod3}$ Region to Determine $\Delta v_E$:}
    \[
    \sum{v_{(x_2,x_1,x_0)}} = -2 \equiv 1 \pmod{3} =\Delta v_E
    \]
Thus, there are 3 possible cases for such $V_E = 1/-2$:
\begin{itemize}
    \item $\sum{v_{(x_2,x_1,x_0)}}$ +1 $\equiv$ $1 \pmod{3}$ due to $-1 \to 0$ or $0 \to 1$.
    \item $\sum{v_{(x_2,x_1,x_0)}}$ -2 $\equiv$ $1 \pmod{3}$ due to $1 \to -1$.
    \item There are two or more errors, but we only consider cases with at most one error.
\end{itemize}

    \State \quad \textbf{Check the $v_{x_2}$ region to Determine the $x_2$ coordinate of the E:}
    \[
    \sum v_{(1,x_1,x_0)} + 2 \cdot \sum v_{(1,x_1,x_0)} \equiv -1 + 2 \cdot (-1) = -3 \pmod{3} \equiv 0 \pmod{3}.
    \]
Thus, the mistake is neither located at $(1,x_1,x_0)$ nor $(2,x_1,x_0)$, which must be $(0,x_1,x_0)$.

    \State \quad \textbf{Check the $v_{x_1}$ region to Determine the $x_1$ coordinate of the E:}
    \[
    \sum{v_{(x_2,1,x_0)}} + 2 \cdot \sum{v_{(x_2,2,x_0)}} = 2 + 2\cdot0 = 2 \equiv 2 \pmod{3}.
    \]
As designed, we previously have $\sum{v_{(x_2,1,x_0)}} = A$ and $\sum{v_{(x_2,2,x_0)}} = B$, such that 
    \[
    (A + 2 \cdot B)\equiv 0 \pmod{3}
    \]

Thus, there are four possible cases for $v_E$:
\begin{itemize}
    \begin{minipage}{0.45\textwidth}
    \item $\big((A+2) + 2 \cdot B\big) \equiv 2 \pmod{3}$
    \item $\big((A-1) + 2 \cdot B\big) \equiv 2 \pmod{3}$
    \end{minipage}%
    \begin{minipage}{0.45\textwidth}
    \item $\big(A + 2 \cdot (B+1)\big) \equiv 2 \pmod{3}$
    \item $\big(A + 2 \cdot (B-2)\big) \equiv 2 \pmod{3}$
    \end{minipage}
\end{itemize}

Since $v_{E} = 1/-2$, the mistake is in region $B$, which is $(x_2,2,x_0)$. Combine with the $x_2$ we get above, we now know the error is at $(0,2,x_0)$.

    \State \quad \textbf{Check the $v_{x_0}$ region to Determine the $x_0$ coordinate of the E:}
    \[
    \sum{v_{(x_2,x_1,1)}} + 2 \cdot \sum{v_{(x_2,x_1,2)}} = -2 + 2\cdot0 = -2 \equiv 1 \pmod{3}.
    \]
As designed, we previously have $\sum{v_{(x_2,x_1,1)}} = A$ and $\sum{v_{(x_2,x_1,2)}} = B$, such that 
    \[
    (A + 2 \cdot B) \pmod{3}\equiv 0 \pmod{3}
    \]
Thus, there are four possible cases for $v_E$:

\begin{itemize}
    \begin{minipage}{0.45\textwidth}
          \item   $\big((A+1) + 2 \cdot B\big) \equiv 1 \pmod{3}$
          \item   $\big(A + 2 \cdot (B+2)\big) \equiv 1 \pmod{3}$
    \end{minipage}%
    \begin{minipage}{0.45\textwidth}
            \item  $\big((A-2) + 2 \cdot B\big) \equiv 1 \pmod{3}$
            \item  $\big(A + 2 \cdot (B-1)\big) \equiv 1 \pmod{3}$
    \end{minipage}
\end{itemize}
Since $v_{E} = 1/-2$, the mistake is region $A$, which is $(x_2,x_1,1)$.

\State \textbf{Fix the Mistake:}
    \[
    \begin{bmatrix}
        \bm{1} & \bm{-1} & -1 \\
        \bm{0} & 0 & 1 \\
        1 & \textcolor{green}{+1} & 0 \\
    \end{bmatrix}
    \quad
    \begin{bmatrix}
        \bm{-1} & -1 & 0 \\
        0 & 1 & 0 \\
        -1 & 0 & 1 \\
    \end{bmatrix}
    \quad
    \begin{bmatrix}
        -1 & 0 & 0 \\
        1 & -1 & 0 \\
        0 & 1 & -1 \\
    \end{bmatrix}
    \]
Therefore, the location of the mistake is $(0,2,1)$. Since $v_{(0,2,1)} = -1$, by $v_{mod3}$, we know the correct value should be $+1$.
\end{algorithmic}
\end{algorithm}
So, the correct code supposes to be: \textbf{12}2,\textbf{0}01,110,\textbf{2}20,010,201,200,120,012. Then, after removing all redundant trits on 0th, 1st, 3rd, 9th position, we get the original message: \textbf{20,111,020,010,201,200,120,012}, which is the same as the message we pick for algorithm 1.


\subsection{General Base-Prime SEC--DED Code}
This method works not only for ternary codes, but also for all prime numbers x-base. With n+1 redundant "digits", it supports storage of \( x^{n} - n-1 \) message digits within \( x^{n} \) position in total, enabling correction of up to one error and detection of up to two errors. This results in a linear code of the form $[x^{n},x^{n}-n-1,3]_x$, or equivalently, $[x^{n-1},x^{n-1}-n,3]_x$.
\subsubsection{General Description}
\begin{algorithm}[H]
\caption{General Coding}
\begin{algorithmic}[1]
    \State \textbf{"Assign indices from 0 to $x^n-1$ to the total $x^n$ positions."}
    \State \textbf{Filling \( x^n-n-1 \) message trits in $x^n$ position}:
       \begin{itemize}
       \item Skip the 0th position and all \( x^i \)-th positions for each \( i \in \mathbb{N}_0 \) with \( i < n \); These are reserved for redundant trits.
       \item Place all message trits in the remaining positions.
   \end{itemize}

    \State \textbf{Setting Values for Redundant trits with index \( x^i \)}:
   \begin{itemize}
            \item Call this redundant trit $R_{i}$ with value $v_{x^i}$.
            \item Rewrite \( x^i \) in base \( x \), as $x_{n-1},...,x_{i+1},\textcolor{red}{x_{i}},x_{i-1}...,x_0$ = $0\cdots0\textcolor{red}{1}0\cdots0$.
            \item Set the value $v_{x_i}$ for $R_i$ such that
            
            \[
            \left[ \sum_{l=1}^{x-1} l \cdot \sum  v_{(x_{n-1},\cdots,x_{i+1}, l, x_{i-1},\cdots, x_{0} )}\right]  \mod x = 0.
            \]
            $\sum  v_{(x_{n-1},\cdots,x_{i+1}, l, x_{i-1},\cdots, x_{0} )}$ is the sum of all $x$ with $i^{th}$ place of its index equals to $l$.
            
           (This sum is influence by $v_{x_i}$, since $v_{x_i} = v_{(0,\cdots,0,1,0,\cdots,0)}$ with the coefficient l=1 )
   \end{itemize}
    \State \textbf{Setting Value of the 0th position($v_{modx}$)}:
       \begin{itemize}
       \item Set the value of the 0th position such that $\sum_{k=0}^{x^n-1} v_{k}  \mod  x = 0$

   \end{itemize}
\end{algorithmic}
\end{algorithm}

\subsubsection{Proof}
Here we prove that the code generated by Algorithm 3 is single error correction(SEC).
\begin{proof}
   Assume that there is exactly one "digit" error that has become any other integer from 0 to \( x-1 \), called them $v_i^{\text{flaw}}$ to distinguish with original $v_i$ .\\
First, check the sum of all digits. \\
Recall
    \[
    \left( \sum v_i \right) \mod x = 0
    \]
Assume 
\[
\left( \sum v_i^{\text{flaw}} \right) \mod x = y
\]

where 
\( y \neq 0, y \in \mathbb{Z}^+, y<x \). Thus, we can determine that the change at the erroneous position must be either \( +y \) or \(-(x-y)\) compared with the original correct value.

Next, calculate the bias sum for the \( x^i \)-th region:
          \[
           \left[ \sum_{l=1}^{x-1} l \cdot \sum  v^{\text{flaw}}_{(x_{n-1},\cdots,x_{i+1}, l, x_{i-1},\cdots, x_{0} )}\right]  \mod x = y_i
           \]

Recall that this sum divides all $x^n$ digits into x-1 groups(by the leftmost summation), where each group contains all numbers for which the $i^{th}$ position of index(in base-x) is equal to l(i.e., indices of the form xx...l...xx).\\

For the sake of simplicity, define the sum of each group as $\sigma_j$ and $\sigma_j^{\text{flaw}} = j \cdot \sum  v^{\text{flaw}}_{(x_{n-1},\cdots,x_{i+1}, j, x_{i-1},\cdots, x_{0} )}$,
Recall
    \[
    [1 \cdot \sigma_1 + 2 \cdot \sigma_2 + \cdots + (n-1) \cdot \sigma_{n-1}] \mod x = 0.
    \]
Assume
    \[
    [1 \cdot \sigma_1^{\text{flaw}} + 2 \cdot \sigma_2^{\text{flaw}} + \cdots + (n-1) \cdot \sigma_{n-1}^{\text{flaw}}] \mod x = y_i.
    \]
define $\Delta\sigma_j = \sigma_j^{\text{flaw}} - \sigma_j$
    \[
    [1 \cdot \Delta\sigma_1 + 2 \cdot \Delta\sigma_2 + \cdots + (n-1) \cdot \Delta\sigma_{n-1}] \mod x = y_i.
    \]
Since we have exactly one mistake, such a mistake will be located in one group only.\\
    \[
    \exists! \, j \in \{1, 2, \ldots, n-1\} \ \text{such that} \ \Delta\sigma_j \neq 0
    \]
Let r denote the index of this group. Removing all other terms (where $\Delta\sigma_j=0$), we have:
    \[
    [r \cdot \Delta\sigma_r] \mod x = [r \cdot y] \mod x= y_i.
    \]
Meanwhile, since $\Delta\sigma_r \neq 0$, it is the only error. Then $\Delta\sigma_r = y$ (or $-(x - y)$, as they correspond to the same coefficient and yield the same residue).
Consider the Theorem 2.1(proved on the next page):\\
For \( x \) as a prime, \( x, y, z \in \mathbb{Z}^+ \), \( y, z < x \), there exists exactly one \( l \) such that 
    \[
    l \cdot z \mod x = y,
    \]
where \( l \in \mathbb{Z}^+ \), and \( l < x \).

By this theorem, there are exactly one value of l valid for r, thus this method could detect the error is in which group that split by $x^i$.
so, we could lock the i-th position of index of the error digits under base x. 
Iterate those step for all positions, and we may get the specific location of such error in base-x. Using \( y \), we can get its original value.
\end{proof}

\begin{theorem}
Assume \( x \) is a prime, given \( x, y, z \in \mathbb{Z}^+ \), \( y, z < x \). Prove that there exists exactly one \( l \) such that
    \[
    l \cdot z \mod x = y, 
    \]
for \( l \in \mathbb{Z}^+ \) s.t. \( l < x \).

\end{theorem} 
\begin{proof}
\hspace{0pt} \\
$\Rightarrow$ To prove the uniqueness of $l$, for the sake of contradiction, assume that there exist two \( l \) satisfying such that \( l_1 \cdot z \mod x = y \) and \( l_2 \cdot z \mod x = y \). Thus,
\begin{equation}
\left\{
\begin{aligned}
l_1 \cdot z &= k_1 \cdot x + y \\
l_2 \cdot z &= k_2 \cdot x + y
\end{aligned}
\right.
\end{equation}

for some integers $k_1$ and $k_2$, while $l_1 \neq l_2$. Subtracting them, and rearranging terms to let x be the denominator, we obtain:
\begin{equation}
\begin{aligned}
\frac{(l_1 - l_2) \cdot z}{x} = (k_1 - k_2) \in \mathbb{Z}
\end{aligned}
\end{equation}
Since the quotient of $(l_1 - l_2) \cdot z$ divided by x is an integer, it must be either $x\mid (l_1 - l_2)$ or $x\mid z$. As x is prime, this rules out the possibility of both $(l_1 - l_2)$ and z each contain only part of x's factors.\\

Assume $x\mid (l_1 - l_2)$. Since $l_1,l_2<x$, $|l_1-l_2|<x$. Thus, by definition of divisor, $m\cdot x = l_1 - l_2$, with $\exists m\in\mathbb{Z} \land |m|<1$, which implies that m=0 and $l_1 - l_2 = 0$. 
Contradict with $l_1 \neq l_2$.\\

Now assume $x\mid z$, since $z<x$ and $z \in \mathbb{Z}^+$, this implies that $z=0$, but $0\notin \mathbb{Z}^+$. Thus, there is no z satisfying $x\mid z$.\\

Therefore, we reach a contradict in both cases, two equations in (1) could not hold at the same time with $l_1 \neq l_2$. So, we get that there are at most one $l$.\\

\noindent$\Leftarrow$ To prove existence of $l$: for the sake of contradiction, assume that for some $x,y_0,z$ that $\nexists l$ that satisfy $l \cdot z \mod x = y_0$. Now consider there are (x-1) valid value for $l\in \{1,2,\cdots,x-1\}$. For each $l_i$, there are some $y_i$ such that $l_i \cdot z \mod x = y_i$, and $y_i \neq y_0$. Thus $y_i$ only have (x-2) valid value that $y_i\in \{1,2,\cdots,x-1\}\setminus{y_0}$, by pigeonhole principle, there must be a $y_i$ that corresponding with $l_i$ and $l_j$ that $l_i \neq l_j$, which contradict with $\Rightarrow$.\\

Thus, by combining existence and uniqueness, we concluding that there exists exactly one $l$ satisfying $l\cdot z \mod{x} = y$.
\end{proof}




\section{Parity Check XOR Method}
\subsection{Conception of Pair and XOR-Modulo 3}
Before we introduce the method to apply the ternary parity check using the XOR method, there is an important definition that should be introduced first.

\begin{definition}
\hspace{0pt}

An \textbf{Addition Modulo p} $\oplus_p$ is a type of addition where each trit of the result undergoes a modulo p operation instead of carrying over to the next position. For the sake of simplicity, we refer to it as XOR base $p$ or simply \textcolor{red}{\textbf{XOR}}. For any $A,B\in \mathbb{Z}$, rewritten them in base-p with same length k+1, value of k depends on the length of the one with larger absolute value.\\
    \[
    A = a_k a_{k-1} \cdots a_1 a_0, \quad B = b_k b_{k-1} \cdots b_1 b_0,
    \]
    \[
    A \oplus_p B = C = c_k c_{k-1} \cdots c_1 c_0, \quad \text{s.t. } c_i = (a_i + b_i) \mod p, \quad i=0,1,\ldots,k.
    \]

It satisfies the following properties:
\begin{enumerate}
    \item \textbf{Closure:} For all \( a, b < (p)_{10}^k \), the result of \( a \oplus_p b \) is also less than \( (p)_{10}^k\).
    \item \textbf{Commutativity:} For all \( a, b \in \mathbb{Z} \), \( a \oplus_p b = b \oplus_p a\).
    \item \textbf{Associativity:} For all \( a, b, c \in \mathbb{Z} \), \( (a \oplus_p b) \oplus_p c = a \oplus_p (b \oplus_p c) \).
    \item \textbf{Identity Element:} There exists an element \( 0 \in \mathbb{Z} \) such that for all \( a \in \mathbb{Z} \), \( 0 \oplus_p a = a \oplus_p 0 = a \).
    \item \textbf{Inverse Element:} For each \( a \in \mathbb{Z} \), there exists an element \( b \in \mathbb{Z}\) such that \( a \oplus_p b = b \oplus_p a = 0 \), where \( 0 \) is the identity element. Here we call $a$ and $b$ a \textbf{pair}, and they are \textbf{addition modulo $p$ inverses} of each other. For the sake of simplicity, we refer to them as \textbf{inverse}. With this property, under addition modulo p, we can say that a = -b.
\end{enumerate}
\end{definition}

\begin{note}
\textbf{Special Multiplication:} To simplify the notation, define a multiplication-like operation:
    \[
    r \cdot a \text{ means that } a \text{ is XOR (for its base) with itself } r \text{ times}, \quad r \in \mathbb{N}.
    \]
Typically, \( r \) is written in base-10 format.
\end{note}

\subsection{Application: Index and XOR in Prototype SEC--DED Code}
Parity check can also be performed using the XOR method.
Take the $[27,23,3]_3$ error correction code in section 1 again:
    \[
    \small
    \begin{bmatrix}
        000 & 001 & 002 \\
        010 & 011 & 012 \\
        020 & 021 & 022 \\
    \end{bmatrix}
    \quad
    \begin{bmatrix}
        100 & 101 & 102 \\
        110 & 111 & 112 \\
        120 & 121 & 122 \\
    \end{bmatrix}
    \quad
    \begin{bmatrix}
        200 & 201 & 202 \\
        210 & 211 & 212 \\
        220 & 221 & 222 \\
    \end{bmatrix}
    \xrightarrow{\text{Filled Value}}
    \begin{bmatrix}
        v_{mod3} & v_{x_0} & 2 \\
        v_{x_1} & 0 & 1 \\
        1 & 1 & 0 \\
    \end{bmatrix}
    \quad
    \begin{bmatrix}
        v_{x_2} & 2 & 0 \\
        0 & 1 & 0 \\
        2 & 0 & 1 \\
    \end{bmatrix}
    \quad
    \begin{bmatrix}
        2 & 0 & 0 \\
        1 & 2 & 0 \\
        0 & 1 & 2 \\
    \end{bmatrix}
    \]
Here, instead of calculating regions for $v_{x_0},v_{x_1},v_{x_2}$ separately, we calculating the current $P_{message} = \oplus V_i\cdot(i)_3$.\\
    \[
    P_{message} = \bigoplus_{message} v_i \cdot (i)_3 =
    \begin{aligned}
        & 2 \cdot 002 \oplus 1 \cdot 012 \oplus 1 \cdot 020 \oplus 1 \cdot 021  \oplus 2 \cdot 101 \oplus 1 \cdot 111  \oplus 2 \cdot 120   \\
        & \oplus 1 \cdot 122 \oplus 2 \cdot 200 \oplus 1 \cdot 210 \oplus 2 \cdot 211 \oplus 1 \cdot 221 \oplus 2 \cdot 222
    \end{aligned}
    \begin{aligned}
    = \bf{101}
    \end{aligned}
    \]
Recall $v_{x_0},v_{x_1},v_{x_2}$ have index $001,010,100$. Pick values for them so that:\\
    \[
    P_{all} = 000 = \bigoplus v_i \cdot (i)_3 = \bigoplus_{message} v_i \cdot (i)_3 \oplus (v_{x_0}\cdot 001) \oplus (v_{x_1} \cdot 010)  \oplus (v_{x_2}\cdot 100) = 101 \oplus v_{x_2}v_{x_1}v_{x_0}
    \]
Thus $v_{x_2}v_{x_1}v_{x_0}$ is the inverse of 101, which is 202.\\
    \[
    \begin{bmatrix}
        v_{mod3}=1 & \textcolor{red}{2} & 2 \\
        \textcolor{red}{0} & 0 & 1 \\
        1 & 1 & 0 \\
    \end{bmatrix}
    \quad
    \begin{bmatrix}
        \textcolor{red}{2} & 2 & 0 \\
        0 & 1 & 0 \\
        2 & 0 & 1 \\
    \end{bmatrix}
    \quad
    \begin{bmatrix}
        2 & 0 & 0 \\
        1 & 2 & 0 \\
        0 & 1 & 2 \\
    \end{bmatrix}
    \]
And last step of setting the value for $v_{mod3}$ is exactly the same with usual, set the sum of all value modulo 3(denoted $P_{mod3}$) equals to 0. And now we have a 3*3*3 matrix whose XOR sum has been reset to 000.\\

When exactly one error occurs anywhere, the XOR sum will shift to the value of the index of the error position or the inverse of that index. $P_{mod3}$ then helps determine the correct case. Thus, this provides a more efficient way to explain how our $[3^{r-1},3^{r-1}-r,3]_3$ linear code works.

\section{Refine Code Designs}
Our code can be improved in two distinct directions. The first, termed $\text{A1}$ (Alternative Version 1), removes the special redundant trit $O$ to increase the information rate and support adaptive-length encoding. The second, $\text{A2}$ (Alternative Version 2), introduces an additional special redundant trit $E$ to extend the minimum Hamming distance to 4 (SEC--TED), while also supporting adaptive-length encoding.

\subsection{A1: Higher Information Rate and Adaptive-Length Ternary Code}
Previously, for ternary codes, we only considered the cases where the length of the valid message was in the form of \( 3^n - n - 1 \). Now, utilizing the properties of pairs, we can also allow lengths that do not follow this form. In other words, we can construct a \([ \frac{3^{r}-1}{2}, \frac{3^{r}-1}{2} - r, 3 ]_3\) linear code( r is the number of the redundant trits) A1 that accommodates messages of arbitrary lengths.

\subsubsection{Coding}
Given any ternary message of length \( m \in \mathbb{Z}^+ \), we need n redundant trits, and there exists some \( n \in \mathbb{Z}^+ \) s.t.:
    \[
    m \in \left( g(n-1), g(n) \right],  \quad \quad \text{where } g(r) = \frac{3^r - 1}{2} - r
    \]
Then, a message of length \( m \) would be encoded into a coded message of length \( m+n \) using the following procedure:
\begin{enumerate}
    \setlength\itemsep{-0.3em} 
    \item Create a matrix of size \( 3^n \).
    \item \textbf{Banish} 0 and any index whose most significant nonzero trit equals 2. This ensures that we retain only one member from each pair (since each must correspond to a number that starts with 1).
    \item For rest of the indices, skip indices with the form \( 3^i \) for each \( i \in \mathbb{N}_0 \) with \( i < n \)(these are reserved for redundant trits); And fill in the message in the remaining entries that have not been banished.
    \item Banish any remaining entries that are not in the form of \( 3^i \) (in case the message is not long enough).
    \item Set the values of the redundant trits at \( 3^i \), so that the XOR sum of the current matrix, $P_{all}$, is zero.
\end{enumerate}

\subsubsection{Example: n=3}
Assume that we want to store a ternary message with length 10: "0,211,112,102". Since it is size $10\in \left(g(2),g(3)\right] = \left(2,1 0\right]$, we pick n=3, and arrange it in the form $3\times3\times3$.
    \[
    \small
        \begin{bmatrix}
            \bcancel{000} & 001 & \bcancel{002} \\
            010 & 011 & 012 \\
            \bcancel{020} & \bcancel{021} & \bcancel{022} \\
        \end{bmatrix}
        \quad
        \begin{bmatrix}
            100 & 101 & 102 \\
            110 & 111 & 112 \\
            120 & 121 & 122 \\
        \end{bmatrix}
        \quad
        \begin{bmatrix}
            \bcancel{200} & \bcancel{201} & \bcancel{202} \\
            \bcancel{210} & \bcancel{211} & \bcancel{212} \\
            \bcancel{220} & \bcancel{221} & \bcancel{222} \\
        \end{bmatrix}
    \rightarrow
    \small
    \begin{bmatrix}
        \blackcell & \bm{v_x} & \blackcell \\
        \bm{v_y} & 0 & 2 \\
        \blackcell & \blackcell & \blackcell \\
    \end{bmatrix}
    \quad
    \begin{bmatrix}
        \bm{v_{z}} & 1 & 1 \\
        1 & 1 & 2 \\
        1 & 0 & 2 \\
    \end{bmatrix}
    \quad
    \begin{bmatrix}
        \blackcell & \blackcell & \blackcell \\
        \blackcell & \blackcell & \blackcell \\
        \blackcell & \blackcell & \blackcell \\
    \end{bmatrix}
    \]
And then assign for those redundant trits.
    \[
    \begin{bmatrix}
        \blackcell & 2 & \blackcell \\
        0 & 0 & 2 \\
        \blackcell & \blackcell & \blackcell \\
    \end{bmatrix}
    \quad
    \begin{bmatrix}
        0 & 1 & 1 \\
        1 & 1\rightarrow \textcolor{red}{0} & 2 \\
        1 & 0 & 2 \\
    \end{bmatrix}
    \quad
    \begin{bmatrix}
        \blackcell & \blackcell & \blackcell \\
        \blackcell & \blackcell & \blackcell \\
        \blackcell & \blackcell & \blackcell \\
    \end{bmatrix}
    \]
\begin{center}
    \textit{Note: Convert 1 on 111 to 0 as potential noise.}
\end{center}
Here we pick values for $v_{x_2}v_{x_1}v_{x_0}$ such that the XOR sum of all trits with their values, $P_{all}$, would equal 000. (Here they should be 002, as the XOR sum of the message with their values, $P_{message}$, equal to 001).\\

Now, convert 1 on 111 to 0 as potential noise. Calculating the $P_{all}$ of the current code; it equals 222. This result could be generated by either +1/-2 happened on 222, or -1/+2 happened on 111. Since we do not have a trit representing 222, it could only be 111 here. So, we do the error correction to convert such 0 on 111 back to 1, among 10 valid trits and 3 redundant trits.\\

\subsubsection{Example: n=4}
Furthermore, if we want to store 11 valid trits(let it be "02,111,121,020", which adds a trit 0 at the end of the previous one). It is size $11\in \left(g(3),g(4)\right] = \left(10, 36 \right]$. We have n=4, and arrange it in the form $3\times3\times3\times3$.\\
    \[
    \small
    \begin{bmatrix}
        \blackcell & \bm{v_{0}} & \blackcell \\
        \bm{v_1} & 0 & 2 \\
        \blackcell & \blackcell & \blackcell \\
    \end{bmatrix}
    \quad
    \begin{bmatrix}
        \bm{v_2} & 1 & 1 \\
        1 & 1 & 2 \\
        1 & 0 & 2 \\
    \end{bmatrix}
    \quad
    \begin{bmatrix}
        \blackcell & \blackcell & \blackcell \\
        \blackcell & \blackcell & \blackcell \\
        \blackcell & \blackcell & \blackcell \\
    \end{bmatrix}
    \]
    \[
    \small
    \begin{bmatrix}
        \bm{v_3} & 0 & \bcancel{1002} \\
        \bcancel{1010} & \bcancel{1011} & \bcancel{1012} \\
        \bcancel{1020} & \bcancel{1021} & \bcancel{1022} \\
    \end{bmatrix}
    \quad
    \begin{bmatrix}
        \bcancel{1100} & 1\bcancel{1101} & \bcancel{1102} \\
        \bcancel{1110} & \bcancel{1111} & \bcancel{1112} \\
        \bcancel{1120} & \bcancel{1121} & \bcancel{1122} \\
    \end{bmatrix}
    \quad
    \begin{bmatrix}
        \bcancel{1200} & \bcancel{1201} & \bcancel{1202} \\
        \bcancel{1210} & \bcancel{1211} & \bcancel{1212} \\
        \bcancel{1220} & \bcancel{1221} & \bcancel{1222} \\
    \end{bmatrix}
    \]
    \[
    \small
    \begin{bmatrix}
        \blackcell & \blackcell & \blackcell \\
        \blackcell & \blackcell & \blackcell \\
        \blackcell & \blackcell & \blackcell \\
    \end{bmatrix}
    \quad
    \begin{bmatrix}
        \blackcell & \blackcell & \blackcell \\
        \blackcell & \blackcell & \blackcell \\
        \blackcell & \blackcell & \blackcell \\
    \end{bmatrix}
    \quad
    \begin{bmatrix}
        \blackcell & \blackcell & \blackcell \\
        \blackcell & \blackcell & \blackcell \\
        \blackcell & \blackcell & \blackcell \\
    \end{bmatrix}
    \]
\begin{center}
    \textit{Note: Those positions filled with index will be used for longer message, but banish for current length.
    }
\end{center}


\subsection{A2: Minimum Hamming Distance 4, SEC-TED Code}
Recall that for the shrink ternary code, it will work as long as there is at most one entry from each pair that has been selected. The minimum Hamming distance of the code is 3, this can be improved if the entries are selected according to certain rules, using the idea of linearly independent in XOR.
\begin{definition}
\hspace{0pt}
\textbf{3-Wise XOR Linearly Independent Set}\\
Let set L is a 3-Wise XOR Linearly Independent Set, if it satisfy the following property:\\
Pick any $A,B,C \in L$, such that $ A\neq B\neq C\neq A$. Then:
    \[
    \forall\, \alpha, \beta, \gamma \in \{0,1,2\}, \quad 
    \alpha A \oplus \beta B \oplus \gamma C = 0 
    \iff 
    \alpha = \beta = \gamma = 0.
    \]
Equivalently, pick any $M\subset L$, and $|M| = 3$, then all elements in M are linearly independent under XOR operation.
\end{definition}
\subsubsection{Choices of Indices and Redundant Trits}
Here is the layout of the $[22,16,4]_3$ linear code A2:\\
\begin{minipage}[t]{0.60\textwidth}
    \[
    \begin{bmatrix}
        \blackcell & \blackcell & \blackcell \\
        \blackcell & \textcolor{red}{0011} & \blackcell \\
        \blackcell & \blackcell & 0022 \\
    \end{bmatrix}
    \quad
    \begin{bmatrix}
        \blackcell & 0101 & \blackcell \\
        \textcolor{red}{0110} & \textcolor{red}{0111} & \blackcell \\
        \blackcell & \blackcell & \blackcell \\
    \end{bmatrix}
    \quad
    \begin{bmatrix}
        \blackcell & \blackcell & 0202 \\
        \blackcell & \blackcell & \blackcell \\
        0220 & \blackcell & 0222 \\
    \end{bmatrix}
    \]
    \[
    \begin{bmatrix}
        \blackcell & 1001 & \blackcell \\
        1010 & 1011 & \blackcell \\
        \blackcell & \blackcell & \blackcell \\
    \end{bmatrix}
    \quad
    \begin{bmatrix}
        1100 & 1101 & \blackcell \\
        \textcolor{red}{1110} & \blackcell & \blackcell \\
        \blackcell & \blackcell & \blackcell \\
    \end{bmatrix}
    \quad
    \begin{bmatrix}
        \blackcell & \blackcell & \blackcell \\
        \blackcell & \blackcell & \blackcell \\
        \blackcell & \blackcell & \blackcell \\
    \end{bmatrix}
    \]
    
    \[
    \begin{bmatrix}
        \blackcell & \blackcell & 2002 \\
        \blackcell & \blackcell & \blackcell \\
        2020 & \blackcell & 2022 \\
    \end{bmatrix}
    \quad
    \begin{bmatrix}
        \blackcell & \blackcell & \blackcell \\
        \blackcell & \blackcell & \blackcell \\
        \blackcell & \blackcell & \blackcell \\
    \end{bmatrix}
    \quad
    \begin{bmatrix}
        2200 & \blackcell & 2202 \\
        \blackcell & \blackcell & \blackcell \\
        2220 & \blackcell & \blackcell \\
    \end{bmatrix}
    \]
    
    \[
    \begin{bmatrix}
    \textcolor{red}{O(Odd Adjust)}, \textcolor{red}{E(Even Adjust)}
    \end{bmatrix}
    \]
\end{minipage}
\hfill
\begin{minipage}[t]{0.35\textwidth}
    \[
    \begin{aligned}
    I_1 &= \{O, 0011, 0101, 1001, 0110, 1010,  \\
        &\quad 1100, 0111, 1011, 1101, 1110\}
    \end{aligned}
    \]

    \[
    \begin{aligned}
    I_2 &= \{E, 0022, 0202, 2002, 0220, 2020, \\
        &\quad 2200, 0222, 2022, 2202, 2220\}
    \end{aligned}
    \]

    \[
    R = \{0011,0110,0111,1110\}\subset I_1
    \]

    \[
    I = I_1 \cup I_2
    \]
\end{minipage}

$I_1\setminus\{O\}$ and $I_2\setminus\{E\}$ are both 3-wise linearly independent sets, which will be proved later with the general $I_1$ and $I_2$ together. 
We select $R = \{0011,0110,0111,1110\}$ as the set of regular redundancy trits.By setting appropriate values for these positions, we can reset the overall XOR sum of the code, $P_{all}$ to 0000. Next, we assign values to the two special redundancy positions $O$ and $E$, such that the sum of the values at indices in $I_1,I_2$, denoted $P_1,P_2$, are both divisible by 3. Further sections will prove that the code guarantees correction of 1 error, detection of 2 errors, and reliable reporting of error presence when 3 errors occur. So, A2 code is a ternary linear block code with block length 22, message length 16, rate 0.727, and minimum distance 4, denoted $\left[22,16,4\right]_3$ code or SEC--TED code.


\subsubsection{Example}

\begin{algorithm}[H]
\caption{Encode $[22,16,4]_3$ Linear Code A2-Initialize: 0,211,001,022,101,122}
\begin{algorithmic}[1]
    \State \quad \textbf{Initialize: } Place the 0,211,001,022,101,122 in message trits, and calculating $P_{message}$:\\
\begin{minipage}[t]{0.60\textwidth}
    \[
    \begin{bmatrix}
        \blackcell & \blackcell & \blackcell \\
        \blackcell & \textcolor{red}{0011} & \blackcell \\
        \blackcell & \blackcell & 0 \\
    \end{bmatrix}
    \quad
    \begin{bmatrix}
        \blackcell & 2 & \blackcell \\
        \textcolor{red}{0110} & \textcolor{red}{0111} & \blackcell \\
        \blackcell & \blackcell & \blackcell \\
    \end{bmatrix}
    \quad
    \begin{bmatrix}
        \blackcell & \blackcell & 1 \\
        \blackcell & \blackcell & \blackcell \\
        1 & \blackcell & 0 \\
    \end{bmatrix}
    \]
    \[
    \begin{bmatrix}
        \blackcell & 0 & \blackcell \\
        1 & 0 & \blackcell \\
        \blackcell & \blackcell & \blackcell \\
    \end{bmatrix}
    \quad
    \begin{bmatrix}
        2 & 2 & \blackcell \\
        \textcolor{red}{1110} & \blackcell & \blackcell \\
        \blackcell & \blackcell & \blackcell \\
    \end{bmatrix}
    \quad
    \begin{bmatrix}
        \blackcell & \blackcell & \blackcell \\
        \blackcell & \blackcell & \blackcell \\
        \blackcell & \blackcell & \blackcell \\
    \end{bmatrix}
    \]
    \[
    \begin{bmatrix}
        \blackcell & \blackcell & 1 \\
        \blackcell & \blackcell & \blackcell \\
        0 & \blackcell & 1 \\
    \end{bmatrix}
    \quad
    \begin{bmatrix}
        \blackcell & \blackcell & \blackcell \\
        \blackcell & \blackcell & \blackcell \\
        \blackcell & \blackcell & \blackcell \\
    \end{bmatrix}
    \quad
    \begin{bmatrix}
        1 & \blackcell & 2 \\
        \blackcell & \blackcell & \blackcell \\
        2 & \blackcell & \blackcell \\
    \end{bmatrix}
    \]
    \[
    \begin{bmatrix}
    \textcolor{red}{O}, \textcolor{red}{E}
    \end{bmatrix}
    \]
\end{minipage}
\hfill
\begin{minipage}[t]{0.35\textwidth}
    \[
    (I_{\text{message}} = I \setminus \bigl(R \cup \{O, E\}\bigr))
    \]
    \[
    P_{message} = \bigoplus_{i \in I_{\mathrm{Message}}} = v_i \cdot i
    \]
    \[
    \begin{aligned}
    \quad = 2 \cdot 0101
            \oplus 1 \cdot 0202
            \oplus 1 \cdot 0220\\
    \quad  \oplus 1 \cdot 1010 
            \oplus 2 \cdot 1100
            \oplus 2 \cdot 1101\\
    \quad  \oplus 1 \cdot 2002
            \oplus 1 \cdot 2022 
            \oplus 1 \cdot 2200\\
    \quad \oplus 2 \cdot 2202 
      \oplus 2 \cdot 2220 = \textbf{1202}\\
    \end{aligned}
    \]
\end{minipage}
\end{algorithmic}
\end{algorithm}

\begin{algorithm}[H]
\caption{Encode $[22,16,4]_3$ Linear Code A2-Set Redundant: 0,211,001,022,101,122}
\begin{algorithmic}[1]
\State \quad\textbf{Set the Redundant Trits:} 
Fill in the value for redundancy trits to reset the XOR sum of the code.
\begin{minipage}[t]{0.60\textwidth}
    \[
    \begin{bmatrix}
        \blackcell & \blackcell & \blackcell \\
        \blackcell & \textcolor{red}{2} & \blackcell \\
        \blackcell & \blackcell & 0 \\
    \end{bmatrix}
    \quad
    \begin{bmatrix}
        \blackcell & 2 & \blackcell \\
        \textcolor{red}{0} & \textcolor{red}{2} & \blackcell \\
        \blackcell & \blackcell & \blackcell \\
    \end{bmatrix}
    \quad
    \begin{bmatrix}
        \blackcell & \blackcell & 1 \\
        \blackcell & \blackcell & \blackcell \\
        1 & \blackcell & 0 \\
    \end{bmatrix}
    \]
    \[
    \begin{bmatrix}
        \blackcell & 0 & \blackcell \\
        1 & 0 & \blackcell \\
        \blackcell & \blackcell & \blackcell \\
    \end{bmatrix}
    \quad
    \begin{bmatrix}
        2 & 2\rightarrow\textcolor{green}{0} & \blackcell \\
        \textcolor{red}{2} & \blackcell & \blackcell \\
        \blackcell & \blackcell & \blackcell \\
    \end{bmatrix}
    \quad
    \begin{bmatrix}
        \blackcell & \blackcell & \blackcell \\
        \blackcell & \blackcell & \blackcell \\
        \blackcell & \blackcell & \blackcell \\
    \end{bmatrix}
    \]
    \[
    \begin{bmatrix}
        \blackcell & \blackcell & 1 \\
        \blackcell & \blackcell & \blackcell \\
        0 & \blackcell & 1 \\
    \end{bmatrix}
    \quad
    \begin{bmatrix}
        \blackcell & \blackcell & \blackcell \\
        \blackcell & \blackcell & \blackcell \\
        \blackcell & \blackcell & \blackcell \\
    \end{bmatrix}
    \quad
    \begin{bmatrix}
        1 & \blackcell & 2 \\
        \blackcell & \blackcell & \blackcell \\
        2 & \blackcell & \blackcell \\
    \end{bmatrix}
    \]
    \[
    \begin{bmatrix}
    \textcolor{red}{O=2}\rightarrow \textcolor{green}{1}, \textcolor{red}{E=0}
    \end{bmatrix}
    \]
\end{minipage}
\hfill
\begin{minipage}[t]{0.35\textwidth}
    \[
    (R = \{0011,0110,0111,1110\})
    \]
    \[
    P_{redundant} = \bigoplus_{i \in R} v_i \cdot i =\bf{2101}
    \]
    \[
    \begin{aligned}
        \quad = 2 \cdot 0011
        \oplus 0 \cdot 0110
        \oplus 2 \cdot 0111 
        \oplus 2 \cdot 1110
    \end{aligned}
    \]
    \[
    P_{all} = \bigoplus_{i \in I} v_i \cdot i 
    \]
    \[
    \begin{aligned}
    \quad = P_{redundant} \oplus P_{message}\\
    \quad = 2101 \oplus 1202 = \bf{0000}
    \end{aligned}
    \]
\end{minipage}
\begin{center}
    \textit{Note: Convert 2 on 1101 to 0, and convert 2 on O to 1 as potential noises.}\\
\end{center}
\State \quad Further, set the $O$ and $E$:
    \[
    P_1 = (\sum_{i\in I_1-O}i+O)\mod{3} = (13+O)\mod{3}=0= (\sum_{i\in I_2-E}i+E)\mod{3} = (9+E)\mod{3} = P_2
    \]
Thus O=2,E=0, set them at the end. We successfully construct the code: 2,020,211,001,022,210,112,220.

\end{algorithmic}
\end{algorithm}

\begin{algorithm}[H]
\caption{Decode $[22,16,4]_3$ Linear Code A2: 2,020,211,001,022,210,112,220}
\begin{algorithmic}[1]
\State \quad \textbf{Noise Generated:} We modify 2 arbitrary trits here as potential noises, 2 on 1101 to 0, and 2 on O to 1: 2,020,211,001,02\textcolor{green}{2},210,112,2\textcolor{green}{1}0. 
\State \quad \textbf{Check 3 Features:}
Check the XOR product sum, the odd indices value sum, the even indices value sum.
    \[
    P_{all} = \bigoplus_{i\in I} v_i\cdot i = 1101, \quad P_1 = \sum_{i\in I_1} v_i\mod{3} = 0, \quad P_2 = \sum_{i\in I_2} v_i\mod{3} = 0
    \]

\State \quad \textbf{Conclusion:} Since both odd and indices value sums are equals to 0(after modulo 3), but the XOR sum of the code is not. This imply that there are at least 2 error occur in either odd or even region. Therefore, we detect that \textbf{there are at least two mistakes}.

\end{algorithmic}
\end{algorithm}


\subsubsection{Extended A2 Codes and Adaptive-Length Construction}

The idea of A2 code could also be generalized as a class of ternary linear codes. Define f'(r) as the cardinality of the maximum 3-wise linearly independent set under length r, it will be $[2{f('r)}+2,2{f'(r)}-r,4]_3$. There are r regular redundant trits, 1 odd redundant trit, and 1 even redundant trit. (The following \( f(r) \) still needs to be proven equal to the previously defined \( f'(r) \).)\\
\begin{minipage}[c]{0.3\textwidth}
\centering
    \[
    n =
    \begin{cases}
        2, & \text{if } r = 3, \\
        \left\lfloor \frac{r}{2} \right\rfloor, & \text{if } 4 \leq r \leq 7, \\
        \left\lceil \frac{r}{2} \right\rceil - 1, & \text{if } r \geq 8.
    \end{cases}
    \]
    \[
    f(r) = \sum_{i=n}^{2n-1} \binom{r}{i} 
    \]
\end{minipage}
\hfill
\begin{minipage}[c]{0.68\textwidth}
\centering
\small 
\begin{tabular}{|r|r|r|r|r|r|r|r|r|r|}
\hline
\textbf{r} & 3 & 4 & 5 & 6 & 7 & 8 & 9 & $\cdots$ \\
\hline
\textbf{f(r)} & 4 & 10 & 20 & 41 & 91 & 182 & 372 & $\cdots$ \\
\hline
\textbf{2f(r)+2} & 10 & 22 & 42 & 84 & 184 & 366 & 746 & $\cdots$ \\
\hline
\textbf{2f(r)-r} & 5 & 16 & 35 & 76 & 175 & 356 & 735 & $\cdots$ \\
\hline
\textbf{Rate} & 0.5 & 0.625 & 0.833 & 0.905 & 0.951 & 0.973 & 0.985 & $\cdots$ \\
\hline
\end{tabular}
\end{minipage}
\begin{table}[h]
\centering
\scriptsize
\setlength{\tabcolsep}{2pt}
\label{tab:originalcompare}
\begin{tabularx}{\linewidth}{@{}%
  >{\centering\arraybackslash}p{.05\linewidth}%
  >{\raggedright\arraybackslash}p{.10\linewidth}%
  >{\raggedright\arraybackslash}p{.08\linewidth}%
>{\centering\arraybackslash}p{.18\linewidth}%
  >{\raggedright\arraybackslash}p{.10\linewidth}%
  >{\raggedright\arraybackslash}p{.50\linewidth}%
}
\toprule
Section & Name & $[n,k,d]_q$ & General $[n,k,d]_q$  & xEC-yED & Feature \\
\midrule
4.2 & Alter 2 & $[22,16,4]_3$ & $[2{f'(r)}+2,2{f'(r)}-r,4]_3$ & SEC-TED& Leverages compact indices to enhance computational efficiency. \\
4.3 & A2 sparse & $[10,6,4]_3$ & $[{f'(r)},{f'(r)}-r,4]_3$ & SEC-TED & Drop $O,E,I_2$, concise version of A2 with higher code rate. \\
\bottomrule
\end{tabularx}
\caption{A2 and A2 Sparse}
\end{table}

\begin{algorithm}[H]
\caption{General A2 Construction}
\begin{algorithmic}[1]
    \State \textbf{Construct $I_1$ and $I_2$}:($I = I_1 \cup I_2$ is the set of selected indices)
       \begin{itemize}
       \item Define the set \( I_1 \) as the set of all \( r \)-trit ternary numbers composed only of the trits 0 and 1, where exactly \( k \) of the trits are 1, and \( k \in [n, 2n-1] \), ($n\geq1$). Then construct $I_2$ by multiply all elements in $I_1$ by 2.
      \item Add O in $I_1$, E in $I_2$
        \end{itemize}
\State \textbf{Decide the Redundant Trits and Message Trits}:   
   \begin{itemize}
    \item \textbf{Case 1: Standard Case (when \( r > 7 \lor r \in O(\text{means odd here})\) )}  
    \begin{itemize}
        \item $R_j = \sum_{i = 1+ \left\lfloor \frac{j}{2} \right\rfloor}^{n + \left\lceil \frac{j}{2} \right\rceil} e_i $ 
        \hspace{10pt} (\(\forall j<r, R_j\in I_1\), since there are n or n+1 trits equals to 1 for $R_j$).
        \item Define \( R = \{R_1,R_2,\cdots, R_r\} \), where \( |R| = r \), which are all of our regular redundant trits.
        \item \( O,E \) are special redundant trits.
        \item The remaining indices in \( I \) are message trits.
    \end{itemize}
    
    \item \textbf{Case 2: Special Case (when \(  r \leq 7\land r \in E\text{(means even here)} \)}  
    \begin{itemize}
        \item $R_j = \sum_{i = \left\lceil \frac{j}{2} \right\rceil}^{n + \left\lfloor \frac{j}{2} \right\rfloor} e_i $ 
        \hspace{10pt} (\(\forall j<r, R_j\in I_1\), since there are n or n+1 trits equals to 1 for $R_j$).
        \item Others are same with case 1.
    \end{itemize}
   \end{itemize}
   
   \State \textbf{Set the values for regular redundant trits}: 
   Under ternary XOR, every elementary vector can be written as a linear combination of the vectors $R_j$.
          \begin{itemize}
        \item \textbf{Case 1: Standard Case (when \( r > 7 \lor r \in O\) )}
            \begin{itemize}
       \item $e_i = R_{2i-1}-R_{2i}$, $i\in[1,n]$
       \item $e_j = R_{2(j-(n+1))+1}-R_{2(j-(n+1))}$, $j\in[n+2,r]$
       \item $e_{n+1} = R_1-\sum^{n}_{i=1}e_i$
    \end{itemize}
    
        \item \textbf{Case 2: Special Case (when \(  r \leq 7\land r \in E \))}
            \begin{itemize}
       \item $e_i = R_{2i}-R_{2i+1}$, $i\in[1,n-1]$
       \item $e_j = R_{2(j-n)}-R_{2(j-n)-1}$, $j\in[n+1,r]$
       \item $e_{n} = R_1-\sum^{n-1}_{i=1}e_i$
    \end{itemize}
       \item Thus, we can decomposed the inverse of $P_{message}$ in to sum of $\text{coefficient}*e_i$, and then rewritten $e_i$ by $R_j$ to decide the value $v_{R_j}$ by coefficient of $R_j$.
   \end{itemize}

    \State \textbf{Set the values for Odd and Even Redundant Trits}:
       \begin{itemize}
       \item Set their value $v_O$ and $v_E$ such that 
       \[ P_1 = (v_O + \sum_{i\in I_1-\{O\}} v_{i}) \mod{3} = 0 = (v_E + \sum_{j\in I_2-\{E\}} v_{j}) \mod{3} = P_2
        \]
   \end{itemize}

    \State \textbf{Notice: If Length of message is not in form $2\cdot f(r)-r$}
       \begin{itemize}
    \item Assume the length of the original message is \( m \). Then, there exists an \( n \) such that \( m \in (f(n-1), f(n)] \). Construct \( I_1 \) and \( I_2 \), but remove the last \( (2\cdot f(n) - n - m) \) indices from \( I_2 \) as banished indices(exclude E). Continue with steps 2, 3, and 4 to construct the code.
   \end{itemize}
   
\end{algorithmic}
\end{algorithm}

\subsubsection{Proof of Minimum Distance 4 for the A2 Code}
A minimum Hamming distance of 4 means that any two valid codewords differ in at least 4 trits. This ensures that A2 code can correct any single-trit error and detect the presence of errors when up to three trits are corrupted.

\begin{theorem}
\hspace{0pt}\\
For any two different ternary messages X,Y in same length, they convert to A2 code X' and Y', of course in same length too. X' and Y' have at least 4 trits different.
\end{theorem} 

Recall that the correct code has the following 3 properties:\\
\begin{minipage}{0.3\textwidth}
    \[P_1 = \sum_{i\in I_1} v_i \mod{3} = 0\]
\end{minipage}
\begin{minipage}{0.3\textwidth}
    \[P_2 = \sum_{i\in I_2} v_i \mod{3} = 0\]
\end{minipage}
\begin{minipage}{0.3\textwidth}
    \[P_{all} = \bigoplus_{i \in I} v_{i} \cdot i = 0\]
\end{minipage}\\

\begin{proof}
\hspace{0pt}\\
Try to prove the distance is at least 4, We can analyze the problem by dividing it into 4 cases:

\begin{enumerate}[leftmargin=*,label=\textbf{Case \arabic*:}]
    \item $X$ and $Y$ differ in exactly 1 trit
    \begin{itemize}
        \item Let the differing position be $A \in I \setminus (R \cup \{\text{O}, \text{E}\})$
    \begin{enumerate}[label=\textbf{Subcase 1.\arabic*:}]
        \item $2A \in R$
        \begin{itemize}
            \item Regular redundant trits differ at $2A$
            \item Original messages differ at $A$
            \item Notice $2A\in R\subset I_1$, thus $A\in I_2$
            \item Then X' and Y' have O and E differences, to satisfy $P_1$ and $P_2$:
            \[
            X'_{\text{O}} \neq Y'_{\text{O}} \quad \text{and} \quad X'_{\text{E}} \neq Y'_{\text{E}}
            \]
            \item Total differences:
            \[
            \underbrace{A}_{\text{msg}} + \underbrace{2A}_{\text{red}} + \underbrace{\text{O}+\text{E}}_{\text{chk}} = 4 \geq 4\quad \text{(Holds)}
            \]
        \end{itemize}
        
        \item $2A \notin R$
        \begin{itemize}
            \item Define $\Delta a \equiv X'_A - Y'_A \pmod{3}$
            \item 3-wise independence guarantees:
            \[
            \nexists B,C \in R\ (B,C\neq 2A\notin R)\text{s.t.}\ \Delta a \cdot A \equiv \beta B \oplus \gamma C \pmod{3} \quad (\beta,\gamma \neq 0)
            \]
            \item Minimum redundant differences: 3
            \item Total differences: 
            \[
            \underbrace{1}_{\text{msg}} + \underbrace{3}_{\text{red}}=4 \geq 4 \quad \text{(Holds)}
            \]
        \end{itemize}
    \end{enumerate}
    \end{itemize}

    \item $X$ and $Y$ differ at 2 trits
    \begin{itemize}
        \item Differing positions $A,B \in I \setminus (R \cup \{\text{O}, \text{E}\})$
        \item Modular differences:
        \[
        \Delta a \equiv X'_A - Y'_A \pmod{3},\quad \Delta b \equiv X'_B - Y'_B \pmod{3}
        \]
        
        \begin{enumerate}[label=\textbf{Subcase 2.\arabic*:}]
            \item $B = 2A$ with $\Delta a = \Delta b$
            \begin{itemize}
                \item Message XOR sums identical $\Rightarrow$ Redundant trits match
                \item WLOG, assume: $A \in I_1 \Rightarrow 2A \in I_2$
                \item As a result, their O and E will not match due to $P_1$ and $P_2$:
                \[
                \underbrace{A,B}_{\text{msg}} + \underbrace{\text{O},\text{E}}_{\text{chk}} = 4 \quad \text{(Holds)}
                \]
            \end{itemize}
            
            \item $B = 2A$ with $\Delta a \neq \Delta b$
            \begin{itemize}
                \item Let $\Delta a \cdot A \oplus \Delta b \cdot B \pmod{3}\equiv d \cdot A$ ($d \in \{1,2\}, d\neq0 \text{ as } \Delta a \neq \Delta b$)
                \item 3-wise independence implies at least 3 redundant trits not match:
                \[
                \nexists M,N \in R \setminus \{A,B\}\ \text{s.t.}\ dA \equiv \beta M \oplus \gamma N \pmod{3}
                \]
                \item Total differences:
                \[
                \underbrace{2}_{\text{msg}} + \underbrace{3}_{\text{red}} = 5 \geq 4 \quad \text{(Holds)}
                \]
            \end{itemize}
            
            \item $B \neq 2A$
            \begin{itemize}
                \item 3-wise independence implies at least 2 redundant trits not match:
                \[
                \nexists M \in R \setminus \{A,B\}\ \text{s.t.}\ \Delta a A \oplus \Delta b B \equiv \gamma M \pmod{3}
                \]
                \item Total differences:
                \[
                \underbrace{2}_{\text{msg}} + \underbrace{2}_{\text{red}} = 4 \geq 4 \quad \text{(Holds)}
                \]
            \end{itemize}
        \end{enumerate}
    \end{itemize}

    \item $X$ and $Y$ differ at 3 trits
    \begin{itemize}
        \item Differing positions $A,B,C \in I \setminus (R \cup \{\text{O}, \text{E}\})$
        \item Modular differences:
        \[
        \begin{aligned}
            \Delta a &\equiv X'_A - Y'_A \pmod{3} \\
            \Delta b &\equiv X'_B - Y'_B \pmod{3} \\
            \Delta c &\equiv X'_C - Y'_C \pmod{3}
        \end{aligned}
        \]
        \item 3-wise independence property implies that their message XOR sum can not be same:
    \[
    \nexists (\Delta a, \Delta b, \Delta c) \in \{1,2\}^3 \ \text{s.t.} \ 
    \Delta a \cdot A \oplus \Delta b \cdot B \oplus \Delta c \cdot C \equiv 0 \pmod{3}
    \]
        \item Minimum differences:
        \[
        \underbrace{3}_{\text{msg}} + \underbrace{\geq1}_{\text{red}} \geq 4 \quad \text{(Holds)}
        \]
    \end{itemize}

    \item $X$ and $Y$ differ at $\geq4$ trits
    \begin{itemize}
        \item Immediate holds:
        \[
        \underbrace{\geq4}_{\text{message}} \geq 4 \quad \text{(Satisfies minimum distance requirement)}
        \]
        \item Note: This holds without considering redundant/checksum trits
    \end{itemize}
\end{enumerate}
\textbf{Existence of a Minimal Distance-4:}\\
Further, we need to show there exist a pair of messages have exact distance 4 with their A2 code , given:\\
\[
X = 0,211,001,022,101,122 \quad X' = 2,020,211,001,022,210,112,220
\]
\[
Y = 0,\textcolor{green}{0}\textcolor{green}{2}1,0\textcolor{green}{2}1,022,\textcolor{green}{0}01,122 \quad Y' = 2,0\textcolor{green}{0}0,2\textcolor{green}{2}1,0\textcolor{green}{2}1,022,2\textcolor{green}{0}0,112,220
\]
X and Y are message in same length, but their A2 code have exact 4 trits different, which means the minimum Hamming distance of the A2 code can not be larger than 4. In concluding, the minimum distance of A2 code is 4.
\end{proof}

\subsubsection{A2 Error Correction Strategies}

Recall that the correct A2 code has following 3 syndromes:\\
\begin{minipage}{0.3\textwidth}
    \[P_1 = \sum_{i\in I_1} v_i \mod{3} = 0\]
\end{minipage}
\begin{minipage}{0.3\textwidth}
    \[P_2 = \sum_{i\in I_2} v_i \mod{3} = 0\]
\end{minipage}
\begin{minipage}{0.3\textwidth}
    \[P_{all} = \bigoplus_{i \in I} v_{i} \cdot i = 0\]
\end{minipage}\\

Let X' denote some A2 code. In the following analysis, we disregard cases involving three or more errors. The corresponding strategies are summarized in \hyperref[tab:error-cases]{Table~\ref*{tab:error-cases}}.

\begin{itemize}
    \item \textbf{4-}: It implies at least one error located in $I_1$ and $I_2$ respectively. So, the code detects that there are two or more errors.
    \item \textbf{-C}: It is impossible to alter $P_{\text{all}}$ from 0 to the current value by editing only one trit, which implies that there are two or more errors.
    \item \textbf{1A}: It is a perfect code. Since our code has minimum distance 4, it required at least 4 errors to change to another perfect code.
    \item \textbf{2A}: $P_2\neq 0$ implies that there is at least one error in $I_2$. While $P_1=0$ either 0 or 2+ errors in $I_1$, since $1+2\geq3$, we only take the cases 0 error in $I_1$. Further, by 3-wise linearly independent property:     
    \[\nexists A,B,C \in I_2\setminus\{E\}\ \text{s.t.}\ A\neq B\neq C\neq A,\ \alpha,\beta,\gamma \in \{1,2\}, \text{and } \alpha A \oplus \beta B\oplus \gamma C \equiv 0,\]
    It requires four or more errors to keep $P_{\text{all}}$ unchanged. Thus, such error must be located at $E\in I_2$, and $P_2$ represents its change from the original value. So we have the true value $\overline{E}: X_{\overline{E}} = (X_E-P_2) \mod 3$.
    \item \textbf{3A}: Similar with the 2A, such error must located on $O\in I_1$, the true value $\overline{O}:X_{\overline{O}} = (X_O-P_1)\mod 3$.

    \item \textbf{1B}: $P_1=P_2=0$ implies there is either 0 or 2+ error in each $I_1$ and $I_2$. As $0\notin I$, with $P_{\text{all}},2P_{all}\neq 0$, it implies that there is at least one error for this code. So, it is not the cases that both $I_1$ and $I_2$ have 0 errors, which concludes with 2 or more errors in this code.

    \item \textbf{Complex Case 1(2B):} 
    Define $f = P_2 \cdot P_{\text{all}}$. If $f \notin I_2$, it implies there are two or more errors.
    If $f \in I_2$, then there is only one error located at $f \in I_2$, with the original value $X_{\overline{f}} =(X_f - P_2) \mod 3$
    
    \item \textbf{Complex Case 2(3B):}
    Define $f = P_1 \cdot P_{\text{all}}$. If $f \notin I_1$, it implies there are two or more errors.
    If $f \in I_1$, then there is only one error located at $f \in I_1$, with the original value $X_{\overline{f}} =(X_f - P_1) \mod 3$

\end{itemize}

\begin{table}[h]
\centering

\renewcommand{\arraystretch}{1.5} 
\setlength{\tabcolsep}{5pt} 
\begin{tabular}{|l|c|c|c|}
\hline
 & {A: $P_{\text{all}}=0$} 
 & {B: $P_{all}\in I \lor 2P_{all}\in I$} 
 & {C: $P_{\text{all}}, 2P_{\text{all}}\notin I\cup\{0\}$} \\ 
\hline

{1: $P_1=0, P_2=0$} & Perfect & \multicolumn{2}{c|}{Two Errors \ } \\ 
\hline

{2: $P_1=0, P_2\neq0$} & $X_{\overline{E}} = X_E - P_2$  & \textcolor{red}{Complex Case 1} & Two Errors \\ 
\hline

{3: $P_1\neq0, P_2=0$} & $X_{\overline{O}} = X_O - P_1$ & \textcolor{red}{Complex Case 2} & Two Errors \\ 
\hline

{4: $P_1\neq0, P_2\neq0$} & \multicolumn{3}{c|}{Two Errors \ } \\ 
\hline
\end{tabular}
\caption{Error Distribution Case Analysis} \label{tab:error-cases}
\end{table}

Since A2 code has a minimum distance of 4, when three errors occur the decoder may detect that an error \textbf{exists}, but mistakenly classify it as a single‑error case (a false correction). By definition, it is still recognized as triple errors detection(TED). A2 code can additionally detect the presence of exactly two errors, but it still cannot correct them.

\subsubsection{Linear Independence among Indices}
The index sets $I_1$ and $I_2$ have the property that any three distinct elements within the same set are linearly independent.
\begin{theorem}
\hspace{0pt}\\
Pick any $A,B,C \in I_1\setminus\{O\}$, such that $ A\neq B\neq C\neq A$. Then:\\
    \[              
    \forall \alpha,\beta,\gamma \in \{1,2\}, \alpha A \oplus \beta B\oplus \gamma C \neq 0
    \]
\end{theorem} 
\begin{proof}
\hspace{0pt}\\
This theorem is equivalent to 
    \[
     \forall \alpha,\beta \in \{1,2\},\forall A,B\in I_1 \setminus\{O\}\,\text{where}\ A\neq B, 
    \]
    \[ 
     \nexists C \in I_1\setminus\{O,A,B\}, \ \alpha A \oplus \beta B \equiv C
    \]
Recall how we construct $I_1\setminus{O}$:
\begin{itemize}
   \item Define the set \( I_1 \) as the set of all \( r \)-trit ternary numbers composed only of the trits 0 and 1, where exactly \( k \) of the trits are 1, and \( k \in [n, 2n-1] \).($n\geq1$)
\end{itemize}
Let A have "a" trits equal to 1. Let B have "b" trits equal to 1. And let A and B have "l" trits overlap, which both equal 1. WLOG, $l\leq a\leq b$, and $a,b\in[n,2n-1]$. Then, considering all linear combinations of A and B.

\begin{itemize}
    \item $A\oplus B = S_1$, $S_1$ have $a+b-l$ trits equal to 1, and $l$ trits equal to 2. Since all regular elements in $I_1$ are composed by only 0 and 1, $S_1\in I_1$ only if l = 0. However, $S_1$ have $a+b-l=a+b-0=a+b\geq2n$ trits equal to 1, which is out of range of $I_1$. Thus, $S_1\notin I_1$.
    
    \item $A\oplus 2B = S_2$, have a-l trits equal to 1, b-l trits equal to 2. If $S_2\in I_1$, $b-l = 0$, and $b=l$. Then $a-l\leq0$, and $a-l\notin [n,2n-1]$, contradict . Thus $S_2 \notin I_1$.
    
    \item $2A\oplus B = S_3$, have b-l trits equal to 1, $a-l$ trits equal to 2. If $S_2\in I_1$, $a-l = 0$, and $a=l$. Then $b-l = b-a\leq n-1$(recall $a,b\in[n,2n-1]$). Thus, $b-l\notin [n,2n-1]$, contradict, the number of 1 is out of the range. As a result $S_3 \notin I_1$.
    
    \item $2A\oplus 2B = S_4$, $S_4$ have l trits equal to 1, and $a+b-l$ trits equal to 2. If $S_4\in I_1$, $a+b-l = 0$, and $a+b = l$. However, $0\leq l\leq a\leq b$ and $a,b\geq n\geq1$, contradict. Thus, $S_4\notin I_1$.
\end{itemize}
Since $S_1,S_2,S_3,S_4\notin I_1$, we conclude with $ \nexists C \in I_1\setminus\{O,A,B\}, \ \alpha A \oplus \beta B \equiv C$
\end{proof}
Note: By similar steps, we can prove for $I_2$: Pick any $A,B,C \in I_2\setminus\{E\}$, such that $ A\neq B\neq C\neq A$. Then:\\
    \[              
    \forall \alpha,\beta,\gamma \in \{1,2\}, \alpha A \oplus \beta B\oplus \gamma C \neq 0
    \]

\subsection{Evaluation}
 Moreover, if we remove the design of $O, E, I_2$, only left with $I_1$ the resulting construction is essentially equivalent to alternative version 1 but with a 3-WXLI set, yielding a $[10,6,4]_3$ linear code with a higher code rate. Generally, it would be a $[f'(r),f'(r)-r,4]_3$ linear code, we could call it A2 sparse.
\begin{itemize}
\item[+] Structural Extensibility: Without $I_2$, A2 sparse is easier to raise the minimum Hamming distances, one simply picks index set that is 4+ wise XOR linearly independent.
\item[+] Code rate: The A2 sparse achieves a higher code rate than the A2. When both use $r+2$ redundant trits, A2 sparse can encode a message of length $f(r+2)-(r+2)$, which is greater than the maximum length of $2\cdot f(r)-r$ for the A2.
\item[--] Cost: For the same message, A2 sparse may require longer indices, (e.g., 10111 vs.\ 111), which increases the computational cost of encoding and decoding.
\end{itemize}
\begin{table}[h]
\centering
\scriptsize
\setlength{\tabcolsep}{2pt}
\label{tab:originalcompare}
\begin{tabularx}{\linewidth}{@{}%
  >{\centering\arraybackslash}p{.05\linewidth}%
  >{\raggedright\arraybackslash}p{.10\linewidth}%
  >{\raggedright\arraybackslash}p{.08\linewidth}%
>{\centering\arraybackslash}p{.18\linewidth}%
  >{\raggedright\arraybackslash}p{.10\linewidth}%
  >{\raggedright\arraybackslash}p{.50\linewidth}%
}
\toprule
Section & Name & $[n,k,d]_q$ & General $[n,k,d]_q$  & xEC-yED & Feature \\
\midrule
2.1 & Prototype & $[9,6,3]_3$ & $[3^{r-1},3^{r-1}-r,3]_3$ & SEC-DED & Derive from binary Hamming code. \\
2.3 & General & N/A & $[x^{n-1},x^{n-1}-n,3]_x$ & SEC-DED & General code for all base-prime.\\
4.1 & A1 & $[13,10,3]_3$ & $[\frac{3^{r-1}}{2},\frac{3^{r-1}}{2}-r,3]_3$ & SEC-DED & Ternary Hamming Code \\
4.2 & A2 & $[22,16,4]_3$ & $[2{f'(r)}+2,2{f'(r)}-r,4]_3$ & SEC-TED& Leverages compact indices to enhance computational efficiency. \\
4.3 & A2 sparse & $[10,6,4]_3$ & $[{f'(r)},{f'(r)}-r,4]_3$ & SEC-TED & Concise version with higher code rate. \\
5.1 & Golay & $[11,6,5]_3$ & N/A & DEC-QED & Perfect Code \\
\bottomrule
\end{tabularx}
\caption{Landscape of Constructed Error-Correcting Codes}
\end{table}
\section{Further Direction: Minimum Distance 4+}
Within the framework of A2 sparse code, an n-WXLI set produces a linear code with minimum distance n+1, and hence can correct $\lfloor\frac{n+1}{2}\rfloor$ error.
\subsection{Brief Example: $n = 4$ and the Ternary Golay Code}
For n=4, a 4-wise XOR-linearly independent set gives a code with minimum distance 5. In the Ternary cases, the classical $[11,6,5]_3$ ternary Golay code\cite{Golay1949} can be viewed as arising from exactly such 4-WXLI set, when its parity check matrix columns are taken as elements of the set. Let \( I \) be a \textbf{4-wise XOR linearly independent set (4-XOR-LI set)}, partitioned into message trits \( M \) and redundant trits \( R \):\\
    \[
    I = \big\{ 00001, 00010, 00100, 01000, 10000, 01122, 10212, 12021, 12102, 22110, 22222\big\},
    \]
    \[
    R = \big\{ 00001, 00010, 00100, 01000, 10000 \big\}, \quad M = I \setminus R = \big\{ 01122, 10212, 12021, 12102, 22110, 22222 \big\}.
    \]
And then, we are going to explain explain Golay Code by idea of 4-XOR-LI set.\\
\textbf{Encoding Steps:}
\begin{enumerate}
  \item Map the message \( 012,\!210 \) to trits in \( M \):
    \[
    \big\{ 01122:0,\ 10212:1,\ 12021:2,\ 12102:2,\ 22110:1,\ 22222:0 \big\}.
    \]
  \item Compute the XOR sum of mapped trits:
    \[
    10212 \oplus 21012 \oplus 21201 \oplus 22110 = 11202.
    \]
  \item Set redundant trits in \( R \) to satisfy \( P_{\text{all}} = 0 \)
    \[
  11202 \oplus 22101 = 00000,\qquad 22101 = 00001\oplus00100\oplus(01000\cdot2)\oplus(10000\cdot2)
  \]
    \[
    \big\{ 00001:1, \ 00010:0,\ 00100:1,\ 01000:2,\ 10000:2 \big\}.
    \]
  \item Final code: \( 10,\!122,\!012,\!210 \).
\end{enumerate}

\textbf{Decoding Sample:} Received codeword \( 10,\!122,\!012,\!2\textcolor{red}{22} \), and calculate syndromes:
    \[
    P_{\text{all}} = 00221 = 1 \times 22110 \oplus 2 \times 22222.
    \]
By the 4-XOR-LI property, the unique decomposition identifies errors at \( 22110 \) and \( 22222 \) with $+1$ and $+2$ respectively from original correct values. Correct by subtracting excess values, we have \( 10,\!122,\!012,\!2\textcolor{green}{10} \).

When the message length to be encoded exceeds 6, the same methodology can be applied: from an index space of length $r \geq 6$ (e.g., $\{000000, \cdots, 222222\}$), a larger 4-wise XOR linearly independent set is constructed. Partitioning this set yields the set $R$, thereby resulting in a larger set $M$ capable of accommodating longer messages. Furthermore, the larger the constructed $n$-wise XOR linearly independent set, the higher the code rate of the resulting error-correcting code. This allows its performance to approach that of an optimal perfect code. However, the problem of constructing a maximum-sized $n$-wise XOR linearly independent set under given constraints remains open, pointing to a clear direction for future research.

\section{Summary and Conclusion}
In this paper, we begin with the prototype $[3^r,3^r-(r+1),3]_3$ SEC-DED code, extending the idea of the parity check to the ternary system. Based on this, we construct a universal SEC-DED coding scheme applicable to all prime-number bases and prove its validity. By introducing the concepts of XOR-Modular 3 and inverse pairs, we improve the efficiency of both encoding and decoding processes. We further refine the prototype code to support the encoding of messages of arbitrary length. Meanwhile, we propose two directions for extension: (1) reducing redundant trits to improve the information rate; and (2) selecting an n-wise XOR linearly independent index set to increase the minimum distance, thereby enabling correction of more errors at the cost of reduced information rate. Finally, we demonstrate that an n-WXLI set can be used to construct linear codes with distance n+1. This approach provides a pathway for designing codes with arbitrarily large minimum distances. Notably, the specific case where d=5 corresponds to the ternary Golay code, which strongly validates our methodology.

Our findings suggest that the idea of parity checks maintains strong performance in ternary systems, leveraging properties of pairs that are not present in binary systems. The techniques developed herein pave the way for further research into the application of modular arithmetic principles and the exploration of alternative encoding strategies. By pushing the boundaries of traditional coding theory, we hope to inspire ongoing inquiry into the realm of non-binary codes and their applications in modern communication technologies.

\section{Supplemental Materials}
\section*{Appendix A: Tightness via a Weight-3 Codeword (Prototype)}
We show that the ternary prototype code of length $27$ with slice-parity checks plus a global-sum check has minimum distance exactly $3$ by explicitly constructing a weight-$3$ codeword.

\begin{lemma}[Weight-3 Codeword for the Ternary Prototype]
\hspace{0pt}\\
Consider the index set $\{0,\dots,26\}$ identified with base-$3$ vectors $\mathbf{i}=(i_2,i_1,i_0)\in\{0,1,2\}^3$. 
Pick any base-$3$ vector $\mathbf{u}\in\{0,1,2\}^3$. Define three positions
\[
  \mathbf{a}=\mathbf{u},\qquad 
  \mathbf{b}=\mathbf{u}+\mathbf{1}\pmod 3,\qquad 
  \mathbf{c}=\mathbf{u}+\mathbf{2}\pmod 3,
\]
where $\mathbf{1}=(1,1,1)$ and addition is component-wise mod~$3$. Let the symbol values at these three positions be $(1,1,1)\in\F_3^3$, and all other positions be $0$. Then the resulting length-$27$ vector is a nonzero codeword of weight $3$.
\end{lemma}

\begin{proof}
\hspace{0pt}\\
By construction, for each digit position $j\in\{0,1,2\}$ the triple of digits across $\mathbf{a},\mathbf{b},\mathbf{c}$ is a permutation of $(0,1,2)$, hence the component-wise sum at position $j$ equals $0\pmod 3$. Therefore, every ``slice'' parity that sums digit-$j$ contributions (with any fixed nonzero scalar weight in $\F_3$) vanishes. 
The global-sum parity also vanishes because the three nonzero symbols sum to $1+1+1\equiv 0\pmod 3$. 
The vector is nonzero and has exactly three nonzero entries, so it is a codeword of weight $3$.
\end{proof}

Consequently, the code has $d_{\min}\le 3$. Since standard arguments for this construction rule out weight-$1$ and weight-$2$ codewords, we conclude $d_{\min}=3$.
\hfill$\square$

\section*{Appendix B: Decoding Pseudocode}
We record clear procedural pseudocode for the single-error-correcting (SEC) decoder of the prototype and a detect-$\ge2$ extension (SEC--TED) for the distance-$4$ A2 Sparse code.

\subsection*{B.1 SEC Decoder for the Prototype \([27,23,3]_3\)}
\begin{algorithm}[H]
  \DontPrintSemicolon
  \caption{SEC decoding (ternary prototype)}
  \KwIn{Received word $r\in\F_3^{27}$.}
  \KwOut{Estimated codeword $\hat{c}$.}
  Compute slice syndromes $S_0,S_1,S_2\in\F_3$ according to the three digit-wise parity checks.\;
  Compute global-sum syndrome $S_\mathrm{all}\in\F_3$.\;
  \uIf{$S_0=S_1=S_2=S_\mathrm{all}=0$}{
    \Return $\hat{c}=r$ \tcp*[r]{no error}
  }
  \Else{
    Let $\hat{\mathbf{i}}=(\hat{i}_2,\hat{i}_1,\hat{i}_0)\gets f_\mathrm{locate}(S_2,S_1,S_0)\in\{0,1,2\}^3$.    \tcp*[r]{Locate error from slice syndromes}
    Set $\hat{e} \gets S_\mathrm{all}$ (the symbol error in $\F_3$).\tcp*[r]{Infer error magnitude from the global sum}
    Flip $r$ at position $\hat{\mathbf{i}}$ by $-\hat{e}$ to obtain $\hat{c}$.\tcp*[r]{Correct the erroneous trit}
    \Return $\hat{c}$.\;
  }
\end{algorithm}

\noindent\emph{Note.} The locator $f_\mathrm{locate}$ is determined by the parity-check matrix: the triple $(S_2,S_1,S_0)$ represents the base-$3$ index of the single error; in algebraic terms, it is the column of $H$ hit by the error.

\subsection*{B.2 SEC--TED Decoder for the A2 Sparse Code}
\begin{algorithm}[H]
  \DontPrintSemicolon
  \caption{SEC--TED decoding (A2 Sparse Code)}
  \KwIn{Received word $r\in\F_3^{n}$.}
  \KwOut{Received word $\hat{c}$ or \textbf{DetectedMultipleErrors}.}
  Compute $S_\mathrm{idxXor}\in\F_3^{n},S_\mathrm{value}\in\F_3$\;
  \uIf{$S_\mathrm{idxXor}=S_\mathrm{value} = 0$}{
    \Return $\hat{c}=r$ \tcp*[r]{no error}
  }
  \uElseIf{$S_\mathrm{idxXor}\neq0$}{
    \uIf{$S_\mathrm{loc} = S_\mathrm{value}\cdot S_\mathrm{idxXor} \in I\setminus\{0\}$}{    
        Let $\hat{\mathbf{i}}=(\hat{i}_{n-1},\cdots,\hat{i}_1,\hat{i}_0)\gets f_\mathrm{locate}(S_\mathrm{loc})\in\{0,1,2\}^n$ \tcp*[r]{Locate error from slice syndromes}             
        Set $\hat{e} \gets S_\mathrm{value}$ (the symbol error in $\F_3$) \tcp*[r]{Infer error magnitude from the global sum}
        Flip $r$ at position $\hat{\mathbf{i}}$ by $-\hat{e}$ to obtain $\hat{c}$ \tcp*[r]{Correct the erroneous trit}
        \Return $\hat{c}$
    }
    \Else{
        \Return \textbf{DetectedMultipleErrors}
    }
  }
  \Else{
    Let $\hat{\mathbf{i}}=(\hat{i}_{n-1},\cdots,\hat{i}_1,\hat{i}_0) = (0,\cdots,0,0)$ \tcp*[r]{Error must located $0\cdots0$ in this case}
    Set $\hat{e} \gets S_\mathrm{value}$ \;
    Flip $r$ at position $\hat{\mathbf{i}}$ by $-\hat{e}$ to obtain $\hat{c}$ \;
    \Return $\hat{c}$
  }
\end{algorithm}
\noindent\emph{Note.} Define $S_{idxXor} = \hat{i}_{n-1}\cdots\hat{i}_1\hat{i}_0$, then locator function $f_\mathrm{locate}(S_{idxXor}) = (\hat{i}_{n-1},\cdots,\hat{i}_1,\hat{i}_0)$.

\section*{Appendix C: References (Background and Context)}
\subsection*{C.1 A1 \& A2 Code Interactive Demo}
Access: \url{https://sltracer.github.io/ECC_Paper_Website_Demo/index_SEC_TED_en.html}
\renewcommand{\refname}{C.2 References List}

\end{document}